\documentclass[%
superscriptaddress,
amsmath,amssymb,
 aps,
pra,
floatfix,
]{revtex4-2}

\usepackage{graphicx}
\usepackage[colorlinks]{hyperref}
\hypersetup{%
	plainpages=true,
	breaklinks=true,
	hypertexnames=false,
	pageanchor=true,
	colorlinks=true,
	linkcolor={blue},
	citecolor={red},
	urlcolor={blue},
	anchorcolor={black}
}

\usepackage[dvipsnames]{xcolor}
\usepackage[normalem]{ulem}
\usepackage{amsthm}
\newtheorem{proposition}{Proposition}[section]
\newtheorem{definition}{Definition}[section]
\usepackage{booktabs}

\newcommand\norm[1]{\lVert#1\rVert}
\DeclareMathOperator*{\argmin}{arg\,min}
\DeclareMathOperator*{\argmax}{arg\,max}
\newcommand{\relmiddle}[1]{\mathrel{}\middle#1\mathrel{}}

\begin{document}
\title{
A genetic algorithm to generate maximally orthogonal frames in complex space
}

\author{Sebastián Roca-Jerat}
\email{sroca@unizar.es}
\affiliation {Instituto de Nanociencia y Materiales de Aragón (INMA), CSIC-Universidad de Zaragoza, Zaragoza 50009, Spain}
\affiliation{Departamento de Física de la Materia Condensada, Universidad de Zaragoza, Zaragoza 50009, Spain}

\author{Juan Román-Roche}
\affiliation {Instituto de Nanociencia y Materiales de Aragón (INMA), CSIC-Universidad de Zaragoza, Zaragoza 50009, Spain}
\affiliation{Departamento de Física Teórica, Universidad de Zaragoza, Zaragoza 50009, Spain}
  
\date{\today}

\begin{abstract}

A frame is a generalization of a basis of a vector space to a redundant overspanning set whose vectors are linearly dependent. Frames find applications in signal processing and quantum information theory. We present a genetic algorithm that can generate maximally orthogonal frames of arbitrary size $n$ in $d$-dimensional complex space.
First, we formalize the concept of maximally orthogonal frame and demonstrate that it depends on the choice of an energy function to weigh the different pairwise overlaps between vectors. 
Then, we discuss the relation between different energy functions and well-known frame varieties such as tight and Grassmannian frames and complex projective $p$-designs. Obtaining maximally orthogonal frames poses a global non-convex minimization problem. We discuss the relation with established numerical problems such as the Thomson problem and the problem of finding optimal packings in complex projective space. To tackle the minimization, we design a hybrid genetic algorithm that features local optimization of the parents. To assess the performance of the algorithm, we propose two visualization techniques that allow us to analyze the coherence and uniformity of high-dimensional frames. The genetic algorithm is able to produce highly-symmetric universal frames, such as equiangular tight frames, symmetric, informationally complete, positive operator-valued measurements (SIC-POVMs) and maximal sets of mutually unbiased bases, for configurations of up to $d=6$ and $n=36$, with runtimes of the order of several minutes on a regular desktop computer for the largest configurations.

\end{abstract}

\maketitle

\section{Introduction}
\label{intro}

At the beginning of every introductory course on linear algebra, one is presented with the concept of a basis of a vector space. A basis is a collection of vectors that are linearly independent and form a spanning set of the vector space. In fact, the dimension of a vector space is defined as the number of elements of a basis. If the vector space is further endowed with an inner product, the condition of linear independence is often superseded by the stronger condition of orthogonality between the basis vectors, thus defining an orthogonal basis.  A consequence of these definitions is that in an inner-product (vector) space of dimension $d$ one can only create sets of at most $d$ orthogonal vectors (which would constitute orthogonal bases). In light of this realization, it is only natural to wonder how to generalize the concept of an orthogonal basis to sets of more than $d$ vectors. These over-spanning sets are typically referred to as frames, which are defined and characterized in frame theory \cite{duffin1952a, daubechies1992ten}. The redundancy of frames provides robustness against errors and is thus useful in the encoding and retrieval of signals, with applications in telecommunications and (quantum) information theory \cite{casazza2013finite}. For instance, equiangular tight frames with $d^2$ vectors provide symmetric, informationally complete, positive operator-valued measures that are optimal for quantum state tomography \cite{renes2004symmetric, scott2006tight, slomczynski2015highly, yoshida2022construction, gu2020fast, jameson2024optimal}.

Several families of frames with different properties can be defined, such as tight and Grassmannian frames \cite{strohmer2003grassmannian}. Many of these properties can be attributed to particular instances of what we call maximally orthogonal frames: frames of $n > d$ vectors that, although not completely orthogonal among themselves, are maximally orthogonal in the sense that their pairwise inner products are minimal.

As orthogonality is a pairwise property, increasing the orthogonality between a given pair of vectors will often come at the expense of decreasing the orthogonality between other pairs of vectors. 
Depending on how much value one assigns to the orthogonality of each pair, one will favor one ``maximally orthogonal'' set of vectors or another. Thus, the first step in this paper is providing a formal definition for maximally orthogonal frames of a complex space of arbitrary dimension. Once armed with this definition, we discuss their properties and establish their relation with notable frame families such as tight and Grassmannian frames, equiangular tight frames \cite{strohmer2003grassmannian}, symmetric, informationally complete, positive operator-valued measurements (SIC-POVMs) \cite{fuchs2017the}, maximal sets of mutually unbiased bases \cite{schwinger1960unitary, ivonovic1981geometrical} and complex projective $t$-designs \cite{delsarte1977spherical, hoggar1982tdesings}. We also show that finding maximally orthogonal frames constitutes, in most cases, a numerical optimization problem. We discuss the relation with established numerical problems such as the Thomson problem \cite{thomson1904on}, the problem of finding optimal packings in complex projective space \cite{jasper2019game} and the problem of finding optimal reference states for quantum classifiers \cite{perezsalinas2020data, rocajerat2024qudit}. Inspired by the use of a genetic algorithm for the Thomson problem \cite{morris1996geneticalgorithm}, we present a genetic algorithm that is capable of generating maximally orthogonal frames of arbitrary size in a complex space of arbitrary dimension.  Then, we visualize and measure the quality of the frames produced by the genetic algorithm. In the cases where the problem of finding maximally orthogonal frames is equivalent to the Thomson problem or the problem of finding optimal packings in complex projective space, we compare the genetic algorithm with state-of-the-art (SOTA) numerical methods.

The remainder of the manuscript is organized as follows. Section \ref{definition} is dedicated to defining maximally orthogonal frames, discussing their properties, and relating them to notable frame families. In Section \ref{predecesors} we discuss related numerical problems. The genetic algorithm is presented in Section \ref{algorithm}. In Section \ref{results}, we visualize and discuss the results obtained with the algorithm and its performance. We end the paper with the conclusions and outlook of our work in Sec. \ref{conclusions}. We also provide technical details and complementary results in the appendices.

\section{Maximally orthogonal frames}
\label{definition}
\subsection{Definition}

The central object of our discussion is a set of vectors $\Phi_{d, n} = \{ |\phi_i\rangle \}_{i=1}^{n} \subset \mathbb C^d$ with $\mathbb C^d$ a finite complex coordinate space of dimension $d$. The canonical sesquilinear inner product of two vectors $|\phi_1\rangle$, $|\phi_2 \rangle$ is denoted $\langle \phi_1 | \phi_2 \rangle$. 
\begin{definition}
    \cite{duffin1952a} A set of vectors $\Phi_{d, n} = \{ |\phi_i\rangle \}_{i=1}^{n} \subset \mathbb C^d$ with $n\geq d$ is a frame for $\mathbb C^d$ if there exist frame bounds $0 < A \leq B < \infty$ such that for every $|v\rangle \in \mathbb C^d$
    \begin{equation}
        A \norm{v} \leq \sum_{i=1}^n |\langle v | \phi_i \rangle|^2 \leq B \norm{v} \,.
        \label{eq:framecondition}
    \end{equation}
    \label{def:frame}
\end{definition}
If one can set $A=B$ in Eq. \eqref{eq:framecondition}, the frame is said to be \emph{tight} \cite{waldron2018introduction}. For $n = d$, a frame is just a basis. If $|\langle \phi_i | \phi_i \rangle | = 1 \; \forall \phi_i \in \Phi_{d, n}$, $\Phi_{d, n}$ is a unit norm frame.
Without loss of generality, we will assume that all vectors are of unit norm in the following. In a vector space of finite dimension, such as the ones considered in this paper, \emph{a frame is essentially a finite spanning set}.

We can formulate an energy associated to a set
\begin{equation}
    E_W(\Phi_{d, n}) = \sum_{i \neq j} W\left(|\langle\phi_i  |\phi_j \rangle|\right) \,,
    \label{eq:generalenergy}
\end{equation}
with $W$, the weighting function, an arbitrary real-valued and increasing function well-defined in the interval $[0, 1]$.
Such an energy grows whenever the \emph{overlap} (the modulus of the inner product) of two vectors increases, ceteris paribus, and thus it will be minimized by a set with large over-all orthogonality. We denote a minimizing set as $\bar \Phi_{d, n}$.
\begin{proposition}\label{prop:MOF}
    A set of unit norm vectors $\bar \Phi_{d, n} = \{ |\phi_i\rangle \}_{i=1}^{n} \subset \mathbb C^d$ that minimizes $E_W$ \eqref{eq:generalenergy} is a unit norm frame.
\end{proposition}
\begin{proof}
    We have to show that $\bar \Phi_{d, n}$ is a spanning set of $\mathbb C^d$. First, if $n=d$ any set that minimizes $E_W$ is an orthogonal basis and thus spans $\mathbb C^d$. In the following we assume $n > d$ and prove the proposition by contradiction. Let us assume that $\bar \Phi_{d, n}$ is not spanning and denote its orthogonal complement as $\bar \Phi_{d, n}^\perp$.
    Note that $W(0) \leq W\left(|\langle\phi_i  |\phi_j \rangle|\right) \leq W(1)$, with equality in the lower bound if and only if $| \phi_i \rangle$ and $| \phi_j \rangle$ are orthogonal to each other. Let us define
    \begin{equation}
        \phi_k^{\not \perp} = \{|\phi_{j \neq k}\rangle \in \bar \Phi_{d, n} | W\left(|\langle\phi_k |\phi_j \rangle|\right) > W(0) \} \,.
    \end{equation}
    Let us choose a vector $|\phi_k\rangle \in \bar \Phi_{d, n}$ such that $\phi_k^{\not \perp}$ is not the empty set. A non-spanning set with $n > d$ vectors must contain al least three such vectors. Then, the frame $\bar \Phi_{d, n}'$ formed by replacing $|\phi_k\rangle$ by $|\phi_k'\rangle \in \bar \Phi_{d, n}^\perp$ in $\bar \Phi_{d, n}$ has energy $E_W(\bar \Phi_{d, n}') = E_W(\bar \Phi_{d, n}) - \epsilon$ with
    \begin{equation}
        \epsilon = 2\sum_{j \in \phi_k^{\not \perp}} \left(  W\left(|\langle\phi_k  |\phi_j \rangle|\right) - W(0) \right) > 0\,,
    \end{equation}
    which implies that $\bar \Phi_{d, n}$ does not minimize $E_W$.
\end{proof}
Therefore we refer to any set that minimizes $E_W$ as a \emph{maximally orthogonal frame} (MOF). Depending on the functional form of the weighting function, $W$, the impact of large and small overlaps on the energy can vary and the MOF will depend on $W$. We denote it $\bar \Phi_{d, n}[W]$. Thus, \emph{there is not a unique definition for a MOF}.

\subsection{Properties of the different maximally orthogonal frames}
\label{subsec:properties_MOF}
\subsubsection{Weighting functions: $p$-frame potential and Riesz $s$-energy}

We consider two main families of weighting functions: for $W(x) = x^{2p}$ the energy becomes a $p$-frame potential \cite{renes2004symmetric, bilyk2021optimal, chen2020on}
\begin{equation}
    {\rm FP}_p(\Phi_{d, n}) = \sum_{i\neq j} |\langle\phi_i  |\phi_j \rangle|^{2p} 
    \label{eq:framepotential}
\end{equation}
and for $W(x) = D^{-s}(x)$, with $D(x)$ the trace or chordal distance 
\begin{equation}
    D(x) = 2 \sqrt{1 - x^2} \,,
    \label{eq:tracedistance}
\end{equation}
the energy becomes a projective Riesz $s$-energy \cite{chen2020on, foundations}
\begin{equation}
    {\rm RE}_{s} (\Phi_{d, n}) = \sum_{i \neq j} \left(2 \sqrt{1 - |\langle \phi_i | \phi_j \rangle|^2}\right)^{-s} \,.
    \label{eq:Rieszenergy}
\end{equation}
To be precise, $D^{-s}(x)$ is not per se a valid weighting function as it is not defined in the limit $x \to 1$. This can be fixed, by imposing that $D^{-s}(1) = +\infty$, which would correspond to taking the limit from the left, or by defining the function by parts, such that $D^{-s}(1-\epsilon \leq x \leq 1) = D^{-s}(1 - \epsilon)$ with $\epsilon$ as a free parameter.
In practice, we do the latter.

In the following, we discuss the relation between the $p$-frame potential and the Riesz $s$-energy and the properties of the frames that minimize them.

\subsubsection{Riesz $s$-energy and asymptotic uniformity}
\label{uniformity}

The main interest of the Riesz $s$-energy is that it allows us to reason in terms of the poppy-seed bagel theorem, which states that the minimal-energy arrangement of $n$ particles constrained to a bounded $D$-dimensional surface and subject to a potential of the form $r^{-s}$, where $r$ is the distance between particles, tends to be uniformly distributed for large $n$ when $s \geq D$ \cite{hardin2004discretizing}. 
To understand this result, one must think physically in terms of short- and long-range interactions. For large $s$, large overlaps (short-range interactions) are weighted more heavily in the energy relative to small overlaps (long-range interactions). This implies that in the minimization process vectors prioritize decreasing the few large overlaps with their nearest neighbors before trying to minimize a large number of small overlaps with distant neighbors. The resulting arrangement is uniform because each vector is mostly concerned with being apart from its nearest neighbors and almost insensitive to the global arrangement of vectors. For small $s$, the opposite is true, vectors prioritize maximizing the number of vectors that they have a small overlap with, which may come at the cost of increasing the overlap with a few nearest neighbors, leading to crowding and lack of uniformity. 

The Riesz $s$-energy is invariant under changes in the global phase of the vectors, so we can identify equivalence classes for all unit vectors differing by a global phase, which essentially defines the complex projective space $\mathbb CP^{d-1}$. A parameterization of $\mathbb CP^{d-1}$ in real coordinates forms a manifold of dimension $2(d-1)$. Following the poppy-seed bagel theorem, we will have to set $s \geq 2(d-1)$ to generate MOFs, $\bar \Phi_{d, n}[s]$, that are asymptotically uniform, in the sense of uniformly distributed on the manifold that parametrizes $\mathbb CP^{d-1}$, for $n \to \infty$ \cite{chen2020on}.

The same reasoning can be applied to the frame potential \eqref{eq:framecondition}. Although we cannot invoke the poppy-seed bagel theorem to produce precise regimes depending on the value of $p$ like we have just done for the Riesz $s$-energy, the argument on the basis of competing short- and long-range interactions remains valid. Small values of $p$ give more relative importance to long-range interactions and thus favor less uniform frames, and vice versa.

\subsubsection{Complex projective $p$-designs and uniformity}

Uniformity is only properly defined in the distribution sense, for $n\to \infty$. In this limit, it is useful to view uniformity as the property of a frame that integrating a function over the Haar measure of the corresponding manifold is the same as averaging the function over the frame elements. For a finite frame, an average over the frame elements will only provide an approximation of the integral. In this context, a frame is said to be a complex projective $p$-design if averaging a polynomial of degree $p$ or less over the frame elements is equivalent to integrating the polynomial over the Haar measure of the corresponding manifold \cite{scott2006tight}. In layman's terms, finite frames cannot be uniform, but the frame vectors can be sufficiently evenly distributed on the manifold that averaging simple enough functions over them is equivalent to averaging over the full manifold.

The $p$-frame potential obeys the Welch bound
\begin{equation}
    {\rm FP}_p(\Phi_{d, n}) \geq \frac{n^2}{\binom{d + p - 1}{p}} - n\,,
    \label{eq:welch}
\end{equation}
with equality when the frame is a complex projective $p$-design \cite{welch1974lower, klappenecker2005mutually}. Thus, the uniformity of a finite frame can be quantified by computing the largest $p$ for which it saturates the Welch bound. This also implies that minimizers of the $p$-frame potential tend to be projective $p$-designs.

\subsubsection{$1$-designs are tight}

The frame potential was introduced by Bennedeto and Fickus, originally with $p=1$, for the study of finite normalized tight frames \cite{benedetto2003finite}. They proved that ${\rm FP}_1$ has several convenient properties: its global minima are tight frames and saturate the Welch bound ${\rm FP}_1(\bar \Phi_{d, n} [p = 1]) = n^2/d - n$, and all its local minima are degenerate, i.e. they are global minima and thus tight frames and $1$-designs. This tells us that we can see tight frames as a particular instance of MOFs as defined in Proposition \ref{prop:MOF}. In fact, Benedetto and Fickus already refer to a set of vectors that minimize the frame potential ${\rm FP}_1$ as maximally orthogonal. Thus, they equate MOFs with tight frames. As argued above, it is natural to relax this equation and consider tight frames as just one case of many MOFs. Just like ${\rm FP}_1$ leads to MOFs that are tight, other weighting functions are equally valid and endow the corresponding MOFs with other useful properties, such as uniformity or low coherence.
We mention in passing that another consequence of Benedetto and Fickus's work is that finding tight frames is an easy task. A local optimization from an arbitrary initial condition suffices to construct tight frames numerically, although constructive analytical methods have also been devised \cite{casazza2012autotuning, cahill2013constructing}. Other weighting functions do not lead to such an easily minimizable energy function, and thus obtaining the corresponding MOFs poses a numerical challenge.

We are now in the position to mention that tight frames are not, just by virtue of being tight, uniform. This is to be expected since they are the minimizers of ${\rm FP}_1$ and thus, generally, just projective $1$-designs, the lowest rank of uniformity for finite frames. Note also that there must exist a tight frame for every degenerate global minima of ${\rm FP}_1$, so a tight frame is not a unique object for a given $d$ and $n$, whereas we identify a uniform frame as a highly-symmetric unique arrangement, modulo orthogonal transformations. In fact, for a given configuration of $d$ and $n$ tight frames form either a manifold or a disjoint union of manifolds \cite{dykema2003manifold}.  

\subsubsection{Coherence: the $p \to \infty$ and $s \to \infty$ limits}
\label{coherence}

For $s \to \infty$, we can write 
\begin{equation}
    \lim_{s\to\infty} ({\rm RE}_s)^{1/s} = m^{1/s} \max_{i \neq j}  \frac{1}{D(|\phi_i \rangle, |\phi_j \rangle)} \,,
\end{equation}
where $m$ is the multiplicity of the largest overlap. Therefore, only the maximum overlap(s) (minimum distance(s)) contribute to the energy and sets that only differ in submaximal overlaps will be degenerate. In this setting, the MOF will be given by 
%
\begin{equation}
    \bar \Phi_{d, n}[s \to \infty] = \argmin_{\{\Phi_{d, n}\}} \left\{ \max_{i \neq j} |\langle \phi_i | \phi_j \rangle| \right\} \,.
    \label{eq:mincoherence}
\end{equation}
Note that the same MOF is obtained by minimizing the $p$-frame potential in the limit $p \to \infty$, the two weighting functions are equivalent in this limit. Equation \eqref{eq:mincoherence} indicates that $\bar \Phi_{d, n}[s \to \infty]$ is the set that minimizes the maximum overlap or, equivalently, that maximizes the minimum (trace) distance \eqref{eq:tracedistance} between vectors. Problems of this sort are typically referred to as \emph{packing problems} \cite{packing}. It is also important to identify the coherence $\mu$ as the quantity that is minimized in Eq. \eqref{eq:mincoherence}
\begin{equation}
    \mu(\Phi_{d, n}) = \max_{i \neq j} |\langle \phi_i | \phi_j \rangle| \,.
    \label{eq:coherence}
\end{equation}
The minimizers of coherence are called \emph{Grassmannian} frames \cite{strohmer2003grassmannian}. Low coherence and ``tightness'' are dual properties in terms of the encoding capabilities of a frame. A tight frame allows ``perfect reconstruction'' of the encoded information, whereas a Grassmannian frame maximizes ``robustness against error'' \cite{haas2017on}. In general, it is impossible to optimize both simultaneously. This is evident from the fact that each is associated with minimizing a different energy function, ${\rm FP}_1$ for tightness and ${\rm RE}_\infty$ or ${\rm FP}_\infty$ for low coherence. 

Similarly to the Welch bound for the $p$-frame potential, there exist bounds for the coherence. Namely,
\begin{itemize}
    \item Bukh-Cox bound \cite{bukh2020nearly} \begin{equation}\label{eq:bukh}
        \mu(\Phi_{d, n}) \geq \frac{(n-d)^2}{n\left[1 + (n-d-1)\sqrt{n-d+1}\right] - (n-d)^2}\hspace{1cm}\mathrm{if}\ n > d
    \end{equation}
    \item Welch-Rankin bound \cite{welch1974lower, rankin1955the} \begin{equation}\label{eq:welchrankin}
        \mu(\Phi_{d, n}) \geq \sqrt{\frac{n-d}{d(n-1)}}\hspace{1cm}\mathrm{if}\ n > d
    \end{equation}
    \item Orthoplex bound \cite{conway2002packing} \begin{equation}\label{eq:orth}
        \mu(\Phi_{d, n}) \geq \frac{1}{\sqrt{d}}\hspace{1cm}\mathrm{if}\ n > d^2
    \end{equation}
    \item Levenstein bound \cite{levenshtein1982bounds} \begin{equation}\label{eq:leven}
        \mu(\Phi_{d, n}) \geq \sqrt{\frac{2n - d(d+1)}{(n-d)(d+1)}}\hspace{1cm}\mathrm{if}\ n > d^2
    \end{equation}
\end{itemize}
Frames that saturate the maximum applicable bound are optimal packings, also known as optimal Grassmannian frames.

\subsubsection{Coherence and degeneracy}
\label{subsubsec:degeneracy}

To understand the geometric properties of Grassmannian frames, let us discuss the particular case of $d=2$, $n=5$. Because vectors in $d=2$ equivalent under global phase changes form $\mathbb CP^1$, which is diffeomorphic to the $2$-sphere $S^2$, they can be visualized on the Bloch sphere. Furthermore, for $d=2$ the trace distance between vectors \eqref{eq:tracedistance} is equal to the Euclidean distance between their corresponding points on the Bloch sphere. For larger $d$ the nice visualization is lost, although some intuition about the topology of $\mathbb CP^{d-1}$ can still be obtained \cite{geometry}. Using the Bloch sphere, we can see that there is a continuum of degenerate frames that minimize the coherence, $\{\bar \Phi_{2, 5}[p \to \infty] \}$, which are illustrated in Fig. \ref{fig:cohdeg}. 
\begin{figure}
    \centering
    \includegraphics[width=0.75\columnwidth]{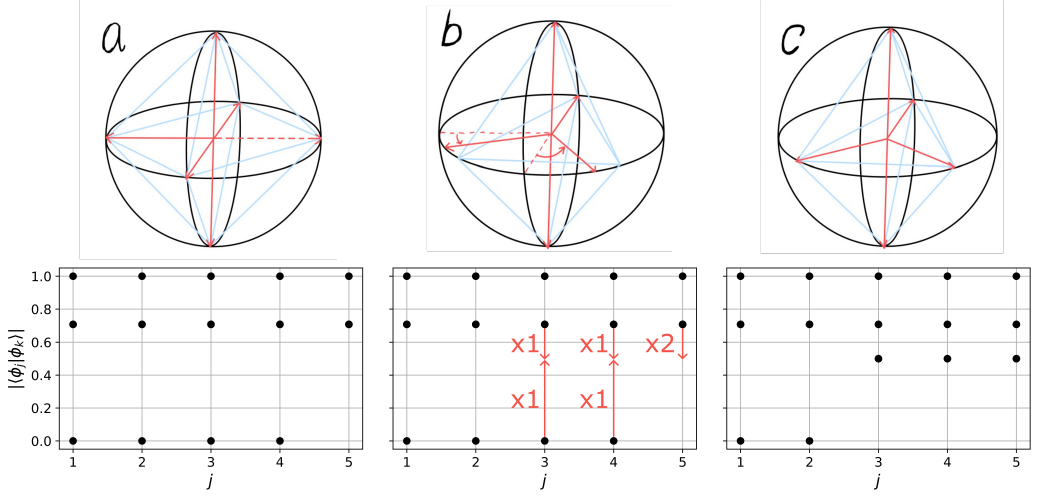}
    \caption{Degeneracy of the coherence. Bloch-sphere representation (top) and overlaps among vectors (bottom) of the different Grassmannian frames for $d=2$ and $n=5$. Column a) shows the distribution of vectors obtained by subtracting one vector from the equator from $\bar \Phi_{2, 6}[s \to \infty]$. The frames depicted in columns b) and c) are obtained by continuously spreading the remaining three vectors at the equator until they form an equilateral triangle. Note that, for each $j$ in all three panels, there are 5 dots in the plot, each corresponding to the overlap of vector $j$ with every vector in the frame (including itself). Visually, there appears to be only three dots for each $j$ because there are only three distinct overlaps and so three of the dots coincide perfectly.}
    \label{fig:cohdeg}
\end{figure}
They are degenerate because their energy is determined by the largest overlap, which is the overlap between the vectors at the Equator of the Bloch sphere and the vectors at the North and South poles. The three vectors around the Equator can be redistributed between forming right angles between themselves, and thus their overlaps among themselves being equal to their overlap with the North- and South-pole vectors, and forming an equilateral triangle. However, this redistribution only decreases the submaximal overlaps, and does not lower the coherence.
Alternatively, one may evaluate the energy of these sets under, e.g., the $6$-frame potential. Because for $p=6$ all overlaps contribute to the energy, the degeneracy is broken and the set in which the vectors at the Equator form an equilateral triangle is lowest in $6$-frame potential. In fact, this set is precisely the set that minimizes the $6$-frame potential and is therefore the MOF for this weighting function, $\bar \Phi_{2, 5}[p = 6]$. Note that, as we just explained, $\bar \Phi_{2, 5}[p = 6] \in \{\bar \Phi_{2, 5}[p \to \infty] \}$. Because the Bloch sphere grants us a graphical interpretation of the different sets, it is easy to see in Fig. \ref{fig:cohdeg} that the vectors in $\bar \Phi_{2, 5}[p = 6]$ are more uniformly distributed on the sphere than the vectors of the other sets in $\{\bar \Phi_{2, 5}[p \to \infty] \}$. The phenomenon that we just described for $d=2$ and $n=5$ is also observed for other values of $d$ and $n$ because it is a consequence of the fact that the coherence disregards submaximal overlaps (other examples appear in Sec. \ref{results}).

Note that our argument in the previous paragraph is centered on the appearance of systematic degeneracies in the coherence. We are disregarding degeneracies arising from the orthogonal symmetry intrinsic to MOFs.  

\subsubsection{Universal maximally orthogonal frames}\label{subsubsec:univarsal_mof}

We have identified tightness and low coherence as associated with different weighting functions and thus generally incompatible in a MOF. We have also discussed the concept of complex projective $p$-designs and the uniformity of frames. Uniformity is to be expected of MOFs that are the minimizers of an energy that favors short-range interactions, without focusing only on nearest-neighbor interactions like coherence does. Thus neither tight nor Grassmannian frames are, generally, maximally uniform given their size. There can exist, however, for certain values of $d$ and $n$, universal MOFs, i.e. frames that are simultaneously tight, minimally coherent and maximally uniform given their size. Equiangular tight frames (ETFs) are the most prominent example \cite{strohmer2003grassmannian}. ETFs are tight frames where every vector has the same overlap with every other vector in the frame.
Their existence is restricted to $n \leq d^2$ \cite{sustik2007on}. ETFs with $n=d^2$ are known in quantum mechanics as symmetric, informationally complete, positive operator-valued measurements (SIC-POVMs). The characteristics that give them their name make SIC-POVMs useful for perfect reconstruction of input states from measurement data \cite{renes2004symmetric, scott2006tight, slomczynski2015highly, yoshida2022construction, gu2020fast, jameson2024optimal}. Their existence for any dimension, $d$ was conjectured by Zauner \cite{zauner2011quantum}. SIC-POVMs are also complex projective 2-designs \cite{renes2004symmetric}. 

The fact that no ETFs exist for $n > d^2$ tells us that $d^2$ marks the threshold beyond which the frame vectors start to crowd the vector space. Since the frame cannot be equiangled while minimizing the energy, some vectors will necessarily have both short- and long-range neighbors. This gives rise to the dynamics afforded by the competition between short- and long-range interactions that we discussed before. It is for $n > d^2$ that tightness, low-coherence and uniformity start to be incompatible. Nevertheless, there do exist highly symmetric frames for $n>d^2$. A notable example are sets of mutually unbiased bases (MUB) \cite{schwinger1960unitary, ivonovic1981geometrical}. A pair of orthonormal bases $\{|\phi_i\rangle\}_{i=1}^d$, $\{|\psi_j\rangle\}_{j=1}^d$ of $\mathbb C^d$ are mutually unbiased if 
\begin{equation}
    |\langle \phi_i | \psi_j \rangle|^2 = \frac{1}{d} \,.
\end{equation}
The maximum number of MUB $\mathcal M(d)$ obeys $\mathcal M(d) \leq d+1$, with equality only if $d$ is an integer power of a primer number. Thus, if $d$ is an integer power of a prime it is possible to construct highly symmetric frames of $d(d+1)$ vectors, by combining the $d+1$ MUB, where every vector has only nearest-neighbors with overlap $1/\sqrt{d}$ (the vectors of the other bases) and next-nearest-neighbors  with null overlap (the other vectors of the same basis). Despite not being ETFs, frames built from $d+1$ sets of MUB are Grassmannian and tight \cite{jasper2019game} and complex projective $2$-designs \cite{klappenecker2005mutually}. 

\section{Related numerical optimization problems}

\label{predecesors}

From the definition of maximally orthogonal frames one can see that finding them constitutes a multivariate minimization problem. Finding the global minima of a non-convex multivariate function is a complex task, without guarantee of success in most cases. A notable exception is the finding of tight frames, ${\rm FP}_1$, where the local minima are all degenerate and thus global minima are commonplace \cite{benedetto2003finite}. In the remaining cases, finding MOFs can be tackled with numerical optimization. In this section, we discuss the relation between the problem of finding MOFs and other established numerical problems.
  
\subsection{The Thomson problem}\label{subsec:thomsonproblem}

The Thomson problem aims to determine the minimum energy arrangement of $N$ electrons constrained to the surface of the unit sphere and subject to electrostatic repulsion \cite{thomson1904on, saff1997distributing}. Thus, if $r_{ij} = \norm{\boldsymbol r_i - \boldsymbol r_j}$ is the distance between each pair of electrons, the Thomson problem for $N$ electrons is equivalent to finding the arrangement $\{\boldsymbol r_i \,| \,\norm{\boldsymbol r_i} = 1 \}_{i=1}^N$ that minimizes
\begin{equation}\label{eq:thomson_energy}
    E = \sum_{i \neq j} \frac{1}{r_{ij}} \,.
\end{equation}
Identifying the electron coordinates on the unit sphere with the coordinates of complex vectors of $d=2$ on the Bloch sphere, it is clear that the Thomson problem is equivalent to the problem of finding MOFs in $d=2$ by minimizing the Riesz $1$-energy.

The Thomson problem is a challenging numerical minimization problem, as the number of local minima is observed to grow exponentially with $n$. In fact, it has been used as a benchmark for global optimization algorithms. There is agreement between all numerical and theoretical methods on what is the global minimum for arrangements of $n \lessapprox 100$ and putatively global minima have also been obtained for $n \lessapprox 1000$  \cite{altschuler2004new, wales2006structure}. Of the different methods employed, a genetic algorithm combined with local optimization at each generation has proven particularly effective, reaching arrangements of $n \leq 200$ \cite{morris1996geneticalgorithm}. The genetic algorithm that we present in the next section can be seen as a generalization of the genetic algorithm employed on the Thomson problem \cite{erber1995comment}.

\subsection{Optimal packings in complex projective space}\label{subsec:optimalpacking}
\label{optimalpackings}

The problem of finding optimal packings in complex projective space consists of finding a set of lines through the origin of $\mathbb C^d$ that are geometrically as spread apart as possible, i.e. with maximum angle distances between them. In practice, the lines are represented by unit vectors, such that a set is represented by a unit norm frame and the geometric spread is quantified by the coherence \eqref{eq:coherence} of the frame.
Since the coherence is invariant under changes in the global phase of the vectors, we can identify equivalence classes for all unit vectors differing by a global phase, so the set of lines in $\mathbb C^d$ is ultimately equivalent to a set of elements of complex projective space $\mathbb CP^{d-1}$. Finding optimal packings in complex projective space is thus equivalent to finding MOFs for ${\rm FP}_\infty$ or ${\rm RE}_\infty$, i.e. Grassmannian frames. 

There exist a number of theoretical and numerical results on Grassmannian frames, which are neatly summarized in the review by Jasper et al. \cite{jasper2019game}. In the review, the authors advertise a website by the name of Game of Sloanes (in honor of Neil Sloane, who hosts a website with the best known packings in Euclidean space) that hosts an open competition to find putatively optimal results for packings in complex projective space \cite{gameofsloanes}.  

The numerical method used to obtain the best known packings listed on the Game of Sloanes consists on sequentially applying two algorithms: local optimization on the Grassmannian manifold \cite{edelman1998geometry, medra2014flexible} followed by alternating projection \cite{tropp2005designing}. As its name implies, the algorithm based on local optimization on the Grassmannian manifold is specific to this problem. On the contrary, alternating projection is rather general and can be used in a variety of settings. In any case it does not work that well on its own, mainly because it benefits from an adequate initial condition, like the one outputted by the local optimization. For this reason, it is interesting to see how a completely general approach like a genetic algorithm fares against these specialized techniques. To this regard, it must be noted that although the search for a Grassmannian frame is rooted in the minimization of the coherence \eqref{eq:coherence}, in practice, to obtain a smooth function, the coherence is defined as the limit 
\begin{equation}
    \mu (\Phi) \equiv \lim_{p \to \infty} \left( {\rm FP}_p(\Phi) \right)^{\frac{1}{2p}}
    \label{eq:smoothcoherence}
\end{equation}
Then, a relatively small value of $p$ is used initially, and several local optimizations are performed sequentially. After each local optimization, the value of $p$ is increased and a new local optimization begins. This process is repeated until convergence is reached \cite{medra2014flexible}. Due to the finite values of $p$ used in practice, the actual numerical minimization is equivalent to the problem of finding MOFs for the p-frame potential with finite $p$. However, in this case, $p$ acts as a numerical hyperparameter that is tuned to optimize results with respect to a figure of merit: the coherence \eqref{eq:coherence}.

\subsection{Optimal reference states in quantum classifiers}
\label{quantumclassifier}

The concept of maximally orthogonal states is introduced in the context of a single-qubit classifier to refer to a set of three or more states in the Bloch sphere that would optimally serve as reference states for each class in the classification process \cite{perezsalinas2020data}. In an $n$-partite classification one would have to select $n$ states of the qubit as reference states, assigning each of them to one of the classes. Then, the training process would strive to learn to assign to each data entry a state on the Bloch sphere that was closest, i.e. with the largest overlap, to the reference state of the corresponding class. This process is facilitated if the reference states are evenly distributed in the state space of the qubit. This is so that the basin of states that are closest to a given reference state is of equal size for all reference states, preventing default biases in the classifier that would have to be overcome during the learning process.  In Sec. \ref{definition} we discussed that uniformity is a property only of some MOFs that are the minimizers of an adequate energy function. However, as we discussed in that same section, for frames with few elements, there exist universal MOFs that achieve tightness, minimal coherence and uniformity. Consequently, in a qubit, and for small sets, coming up with sets of maximally orthogonal states is relatively straightforward, with some imagination it can even be done graphically. This is the case in Ref. \cite{perezsalinas2020data} where only sets of up to six maximally orthogonal states are considered and the challenge of obtaining larger sets is not discussed. The idea is further developed in Ref. \cite{rocajerat2024qudit}, where they study a single-\emph{qudit} classifier as the natural generalization of the single-qubit classifier. In order to generate large sets of optimal reference states of a qudit it was necessary to formalize the concept of maximal orthogonality, where the nuanced distinction between maximum orthogonality and maximum uniformity becomes apparent, and to develop the genetic algorithm that we present in this paper. The equivalence between maximally orthogonal states of a qudit and MOFs is clear from the fact that the coordinates of all possible normalized pure states of a $d$-dimensional qudit form the complex projective space $\mathbb CP^{d-1}$.

\section{The genetic algorithm}
\label{algorithm}

A genetic algorithm is a general optimization method inspired by Darwinian natural selection. A population of candidate solutions to the optimization problem is evolved across time in successive generations. At each generation, the fittest individuals among the population are selected and through recombination, mutation and survival, they generate the next population. The process is iterated until convergence is reached. The process of recombination attempts to stochastically combine the good characteristics of one individual with those of another, giving rise to an over-all better individual. The process of mutation introduces noise, allowing a random exploration of the state space. The process of survival introduces determinism in the algorithm. Because the fittest individuals are allowed to survive, it ensures a steady flow toward better solutions, without stochastic setbacks.
The problem of finding MOFs lends itself to be implemented as a genetic algorithm because the recombination process can be implemented straightforwardly by combining subsets of vectors of the two parent frames. 

In our particular setting, for a given $d$ and $n$, a population, $P_{d, n} = \{\Phi_{d, n}^k\}_{k=1}^N$, will be a collection of putative MOFs, the individuals, with $N$ the size of the population. The vectors in each frame are normalized, $\langle \phi_i^k | \phi_i^k\rangle = 1 \, \forall i, k$, and encoded as a coordinate vector in the orthonormal basis of $\mathbb C^d$. To eliminate degeneracies due to symmetry, we fix the first vector of each frame to be $|\phi_1^k\rangle = (1, 0, \ldots, 0) \, \forall k$. In all other vectors, the global phase is eliminated $|\phi_{i\neq1}^k\rangle = (z_{k, 1}, \ldots, z_{k, d})$, with $z_{k, 1} \in \mathbb R$ and $z_{k, l \neq1} \in \mathbb C$.

Genetic algorithms are typically formulated as a maximization problem, where the fitness $F$ is the function to maximize. Here, we define fitness as the negative energy Riesz $s$-energy: $F(\Phi) = -{\rm RE}_s(\Phi)$. At each generation, we sort the individuals by fitness and select a number of them to serve as parents of the next generation. It is common in genetic algorithms to select the parents stochastically, with a probability proportional to their fitness, to delay the convergence process and efficiently explore the parameter space. Here we employ a rudimentary deterministic alternative. We fix a fitness gap, $\Delta F$, and select four parents in decreasing order of fitness, starting from the fittest individual and ensuring that the next parent has a fitness at least $\Delta F$ smaller than the previous parent. The fitness gap is determined as a fraction of the fitness of the fittest individual in that generation $\Delta F = \alpha_{\rm div} F_{\rm max}$. The diversity ratio, $\alpha_{\rm div}$, is a hyperparameter of the algorithm that we set to $\alpha_{\rm div} = 0.1$.

Recombination is implemented as a random crossover of vectors from two individuals. A trivial recombination could consist on generating a random integer $i$ between $1$ and $n$ to create a child with the $i$ first vectors of parent A and the last $n - i$ vectors form parent B. However, the vectors within a frame are, a priori, not sorted in any particular fashion, i.e. their order does not reflect any geometrical structure in the manifold that parametrizes the vectors.  Therefore, the subset of vectors inherited from each parent would be random, consisting of vectors from parent A that could be arbitrarily close to vectors from parent B, and thus throw away most of the orthogonality gained over the course of the optimization process. 
Because of this, a trivial recombination like the one just described would, most likely, create a child with a fitness much lower than that of its parents. To prevent this, we sort the vectors of each parent by distance to the first vector of each set (which is fixed and the same for all sets). After the sort, the order of the vectors in a set does reflect some structural properties. Namely, the first vector of the set marks a reference point in the manifold that is common for all sets. After sorting, we apply the random recombination by drawing the first $i$ vectors from parent A, which are expected to be relatively closer to the reference point, and the last $n - i$ vectors from parent B, which are expected to be relatively far from the reference point. To better understand this process, let us consider the case $d=2$. Now the manifold is the Bloch sphere and the reference point is given by the vector $(1, 0)$, which marks the north pole. Selecting the first $i$ vectors from parent A corresponds to selecting vectors covering a spherical cap around the north pole with an approximate height determined by the ratio $i/n$. Then, selecting the $n - i$ vectors from parent B corresponds to selecting vectors covering a spherical cap around the south pole with an approximate height determined by the ratio $(n - i)/n$. Although nothing guarantees that the two caps do not overlap or, contrarily, leave blank a spherical strip in between them, at least the geometrical structure within each cap is preserved and the fitness of the child will only depend on the seam between the two caps. A similar recombination strategy is used in Ref. \cite{morris1996geneticalgorithm}.

Following the same reasoning, we implement mutations by sorting a set by distance to the first vector, generating a random integer $i$ between $1$ and $n$ and a random phase $\theta \in [0, 2\pi)$ and rotating the last $n - i$ vectors of the sorted set by the relative phase, i.e. $(z_1, z_2, \ldots, z_d) \to (z_1, z_2 e^{i \theta}, \ldots, z_d e^{i \theta})$. For $d=2$, this amounts to rotating the last $n - i$ vectors of the set an angle $\theta$ around the north pole. 

In addition to the usual operations of a genetic algorithm: selection of the fittest, recombination, and mutation, we also implement local optimization. At each iteration, after the four parents are selected and before recombination and mutation, they are each subjected to a local optimization process. To implement the optimization, we encode a set as a $(n-1) \times d \times 2$ dimensional vector where the first element is the real part of the second element of the first vector, the second element is the imaginary part of the second element of the first vector, the third element is the real part of the third element of the first vector and so on, i.e. $\Phi_{d, n} \to (\Re(z_{2, 1}), \Im(z_{2, 1}), \ldots, \Im(z_{2, d}), \Re(z_{3, 1}), \ldots)$. Due to the necessity to enforce normalization on the vectors composing each set throughout the optimization process, this is a case of multivariate constrained optimization which we tackle using sequential least-squares programming. The local optimization runs until either convergence or a maximum number of iterations, which constitutes another hyperparameter, are reached. From a genetic perspective, applying a local optimization to the parents that is later inherited by the next generation introduces a Lamarckian aspect to an otherwise Darwinian evolution process. Genetic algorithms based on this hybrid approach are typically termed memetic algorithms \cite{moscato2000on}.

In summary, the algorithm runs as follows:
\begin{enumerate}
    \item Select energy function, i.e. select the fitness function as the negative of the energy function.
    \item Generate initial population. We generate $N$ sets of $n$ normalized random vectors drawn from the Haar distribution.
    \item Select the four ``fittest'' individuals as the parents (respecting the selected fitness gap, $\Delta F$).
    \item Apply local optimization to the parents.
    \item Generate twelve children by recombining every parent with every other parent (note that each pair of parents mates twice, exchanging the roles of parents A and B in each case).
    \item Generate four other children by mutating each parent once.
    \item Form a new population with the four parents and the sixteen children and evaluate the fitness of the fittest individual.
    \item If the maximum number of iterations or the convergence criterion are reached, halt, otherwise go to step 3 with the new population.
\end{enumerate}
To induce the halting of the algorithm, one can fix a maximum number of iterations or a convergence criterion based on the rate of change of the fitness across generations.

We arbitrarily fix the number of parents to four, the number of recombinations to twelve, the number of mutations to four and consequently the population size to twenty (inherited from Ref. \cite{morris1996geneticalgorithm}), these numbers could be left as hyperparameters.

More numerical details can be found in App. \ref{app:numerics}

\section{Results and discussion}
\label{results}

\subsection{The Thomson problem}

\begin{figure}
    \centering
    \includegraphics[width = 0.75\columnwidth]{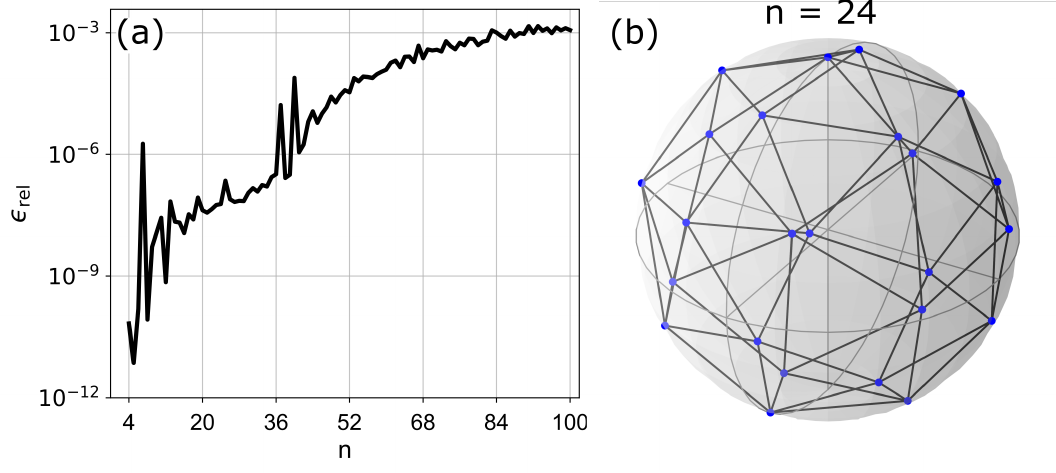}
    \caption{Results for the Thomson problem. a) Relative difference between the energy \eqref{eq:thomson_energy} for the frames obtained with the GA and the energies reported in Ref. \cite{wiki:Thomson_problem} up to $n = 100$. b) For illustration purposes, the Bloch-sphere representation of the MOF obtained with the GA for $n=24$ is shown, which forms a snub cube. Each blue dot represents an electron and the black lines connect nearest neighbors.}
    \label{fig:ThomsonEnergy}
\end{figure}

As presented in Sec. \ref{subsec:thomsonproblem}, for $d=2$ and using the Riesz $1$-energy as the weighting function, the problem of finding MOFs is equivalent to the Thomson problem. Therefore, we use the Thomson problem as a benchmark for the GA. Figure \ref{fig:ThomsonEnergy}(a) shows the relative energy difference, $\epsilon_{\rm rel} = |{\rm RE}_1^{\rm GA} - {\rm RE}_1^{\rm ref}|/|{\rm RE}_1^{\rm ref}|$, between the frames produced by the GA for the Thomson problem and the best known arrangements reported in Ref. \cite{wiki:Thomson_problem} for configurations of up to $n = 100$ points. The GA produces optimal or near-optimal frames with a relative error under $\sim0.1\%$ for $n \leq 100$. 

\subsection{General maximally orthogonal frames}

As discussed in Sec. \ref{definition} the definition of MOF is not unique and depends on the choice of weighting function. Different choices are associated with different properties for the associated MOF. We have discussed three properties that one may seek in a frame: tightness, low coherence and uniformity, which are generally incompatible. Since obtaining tight frames is numerically trivial, here we focus on obtaining frames with low coherence and high uniformity. 

It is important to realize that just because we associate certain properties to the minimizers of certain energy functions, this does not mean that the best way to numerically obtain a frame with a given property is to minimize the corresponding energy function. In particular, our testing on the GA has revealed that minimizing the coherence is not the best way to yield frames with low coherence. The minimization process tends to converge to local minima that are far from optimal. As discussed in Sec. \ref{subsubsec:degeneracy}, this is because the coherence disregards any submaximal overlaps between vectors. Accordingly, in the optimization process, there is no incentive to reposition vectors of a frame that are only involved in submaximal overlaps. This is true even if these vectors are far from their optimal position required to achieve minimal coherence. We see this phenomenon as akin to the vanishing gradient problem in neural networks. When the energy function becomes, or in this case is (by construction), insensitive to some of the optimization parameters, the optimization process is unable to tune these parameters, even if doing so sufficiently would eventually result in a lower energy and thus better performance. In contrast, minimizing a softer energy function, such as the Riesz $s$-energy or the $p$-frame potential with finite $s$ and $p$, which in principle offer no guarantee of leading to low-coherence (Grassmannian) frames, yields better results.
With this, the theoretical one-to-one relationship between energy functions and properties gives way, in practice, to a situation where the choice of energy function is part of the tuning of the algorithm. The parameters $s$ or $p$ become hyperparameters of the GA, that can be tweaked in the search for frames that maximize certain properties, such as uniformity or coherence (See Sec. \ref{subsec:optimalpacking} for a discussion of how this is also the case in SOTA algorithms for the search of optimal packings). 

We do not undertake such systematic hyperparameter optimizations. Following the poppy-seed bagel theorem (See Sec. \ref{uniformity}), we select the Riesz $s$-energy with $s=2d$ as the energy function to minimize by the GA, and we evaluate the results in terms of tightness, coherence and uniformity. We find that with this choice the GA obtains highly symmetric frames with optimal or near optimal coherence values that are in many cases tight and reasonably uniform. We have tested up to $d=7$ and find that performance starts dropping for $d > 5$ and $n > 40$. Further optimization with respect to a particular performance metric, such as coherence or uniformity, should be possible by tweaking $s$ or switching to a $p$-frame potential.

In the following, first we define the metrics and visualization methods that we use to assess the performance of the GA and then we discuss the results.

\subsubsection{Performance metrics and visualization methods}
\label{subsec:metricsbounds}

To evaluate the performance of the GA, we monitor several metrics. Most trivially, we measure the coherence \eqref{eq:coherence} and its distance to the applicable coherence bound (See Sec. \ref{coherence}). A frame is an optimal packing (optimal Grassmanian frame) if it saturates the bound. Similarly, to quantify whether a frame is tight, we make use of the Welch bound \eqref{eq:welch} to define the looseness of a frame as 
\begin{equation}
    \lambda(\Phi_{d, n}) = \mathrm{FP}_1(\Phi_{d, n}) - (n^2/d - n) \,.
    \label{eq:looseness}
\end{equation}
A frame is tight if its looseness is zero. We also check if the Welch bound \eqref{eq:welch} is saturated for $p > 1$ to qualify whether a frame is a $p$-design. Complementarily, we quantify the uniformity of a frame by measuring the mesh norm. The mesh norm of a frame is defined as
\begin{equation}
\label{eq:meshnorm}
    \rho(\Phi_{d, n}) = \min_{|\psi\rangle\in\mathbb C^d}\max_{k<1<N}|\langle\psi|\phi_k\rangle| \,.
\end{equation}
It measures the radius of the largest gap left by vectors of the frame in the vector space. Given the inverse relation between overlap and angle distance, a small mesh norm indicates a large radius. For two frames of the same $d$ and $n$ the larger mesh norm corresponds to the more uniform frame. For reference, note that in the example of Fig. \ref{fig:cohdeg} frame a) has a mesh norm of $1/\sqrt{2} \approx 0.707$, while frame c) has a mesh norm of $\approx 0.852$.

Given the geometric and symmetry properties expected of MOFs, it is also interesting to devise visualization techniques that reveal such properties.
For $d=2$, the Bloch sphere provides a way to visually assess the frames (See Figs. \ref{fig:cohdeg} and \ref{fig:ThomsonEnergy}).
In this fashion, it is possible to visualize that in certain high symmetry MOFs the vectors of the frame correspond to the vertices of certain platonic solids and other emblematic polyhedra.
In higher dimensions, there is no clear way to visualize the full manifold and the coordinates of the frame vectors on it. Instead, we explore here two alternative visualization methods.

The first consists of plotting the overlaps of each vector of the frame with every other vector. Figures \ref{fig:cohdeg} and \ref{fig:overlaps_bestE} employ this method. It clearly reveals the presence or absence of certain symmetries in the arrangement of vectors. In many cases, it possible to discern the angle set of the frame, which is the number of distinct overlaps between the different vectors of the frame. Besides, it is trivial to relate this plot to the coherence of the frame, since the coherence is given by the largest overlap (excluding the one overlap that each vector has with itself). 
In a sense, this type of plot is a visual proxy for coherence.

Alternatively, in Appendix \ref{app:share} we employ another method that is a visual proxy for uniformity. In this method, we measure and plot the share of the vector space that is closest to each vector of the frame. To formalize this concept, we define $\beta_j$ as the normalized share of $\mathbb CP^{d-1}$ for the vector $|\phi_j\rangle \in \Phi_{d,n}$:
\begin{equation}
    \label{eq:share_hilbert}
    \beta_j = n \alpha_j \,,
\end{equation}
with $\alpha_j = |\Psi_j|/|\mathbb CP^{d-1}|$ and
\begin{equation}
    \Psi_j = \left\{ |\psi\rangle\in\mathbb CP^{d-1} \relmiddle| |\phi_j\rangle = \argmax_{|\phi\rangle\in\Phi_{d,n}} |\langle\psi|\phi\rangle| \right\} \,.
\end{equation}
Note that $\alpha_j$ is the normalized cardinality of the set $\Psi_j$ composed of all vectors $|\psi\rangle \in \mathbb CP^{d-1}$ that have more overlap with (are closer to) $|\phi_j\rangle$ than with (to) any other $|\phi_{k\neq j}\rangle\in\Phi_{d,n}$.
A value of $\beta_j \sim 1 \ \forall j$ indicates that the vectors of the frame are evenly distributed in the vector space.
Note that it is also possible to define a metric associated with this visualization method by computing the standard deviation $\sigma_\beta$ of the set $\{\beta_j\}_{j=1}^n$. This provides a metric of uniformity complementary to the mesh norm. See App. \ref{app:numerics} for clarification of how to numerically compute the mesh norm and the share of the vector space.

\subsubsection{Discussion of the results}

\begin{figure}
    \centering
    \includegraphics[width = \textwidth]{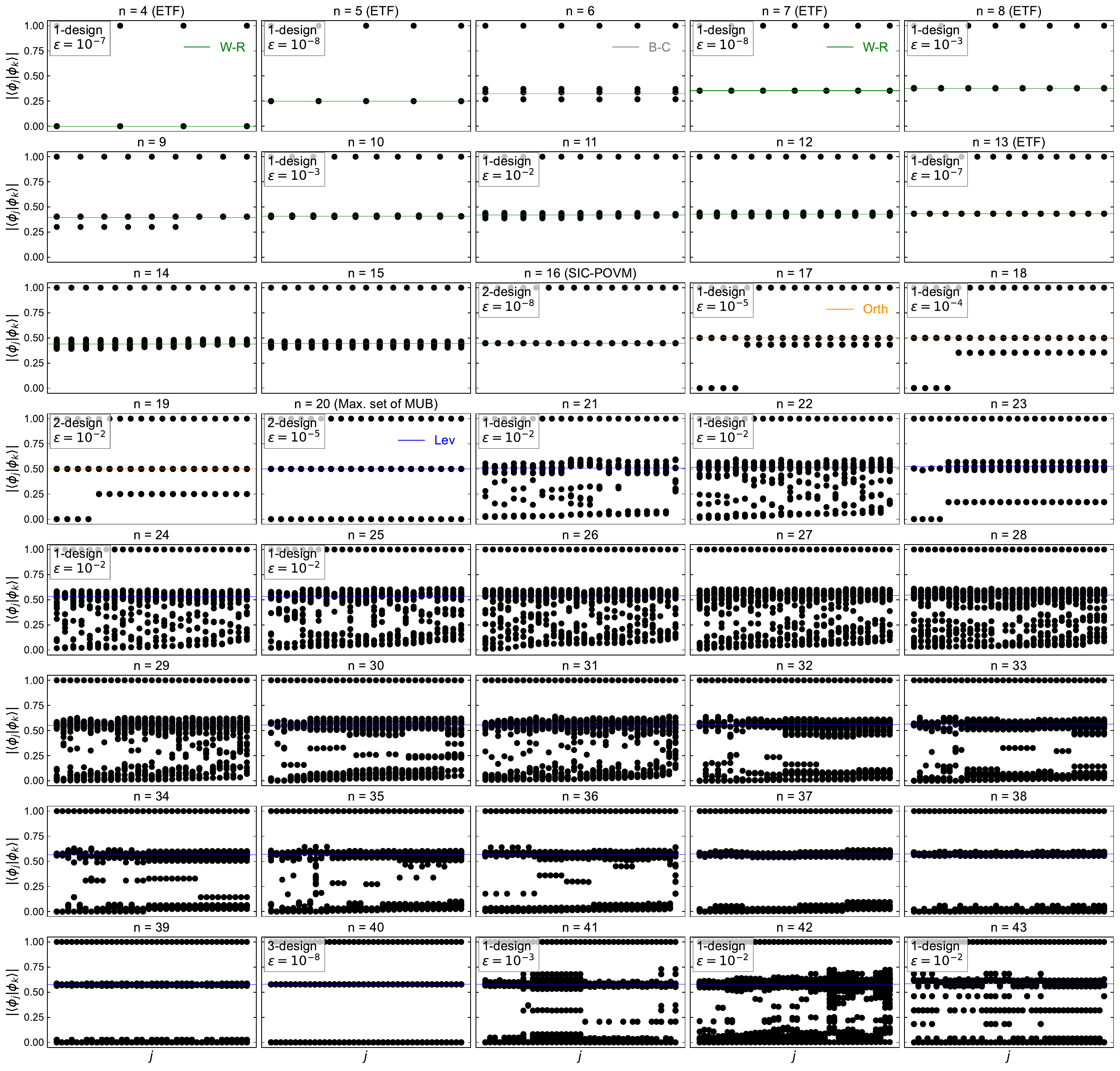}
    \caption{Maximally orthogonal frames for $d=4$ and $s=8$. Each panel shows the overlaps among vectors of the best MOF out of ten runs of the GA for a certain size $n$. The horizontal colored lines indicate the applicable coherence bound. The panel title indicates whether it is a high-symmetry frame such as an ETF, a maximal set of MUB or a SIC-POVM. A box inside the panel indicates whether the frame is a $p$-design and the tolerance $\varepsilon$ with which the corresponding Welch bound \eqref{eq:welch} is saturated.
    }
    \label{fig:overlaps_bestE}
\end{figure}

To illustrate the capabilities of the GA, we focus on the case $d=4$. In Figure \ref{fig:overlaps_bestE}, we show the frames obtained for $n \in [4, 43]$. For each configuration, we plot the result with lowest energy out of ten independent runs of the GA. As explained before, each panel displays the overlap of each vector $|\phi_k\rangle\in\bar\Phi_{4,n}[s=8]$ with all the others including itself. Additionally, in each panel, we mark the applicable coherence bound with a colored horizontal line, we indicate whether the frame is an ETF or a maximal set of MUB in the title, and, in the cases where the frame saturates the Welch bound for some $p$, we label the frame as a $p$-design in a box within the panel. The box also includes the tolerance $\varepsilon$ with which the bound is saturated. We do not label a MOF as a $p$-design unless the tolerance is $\varepsilon \leq 10^{-2}$.

This visualization method helps identify frames with high symmetry. For instance, for an ETF only two rows of overlaps are present: one with unit value corresponding to the overlap of each vector with itself, and another corresponding to a uniform band of value equal to the Welch-Rankin bound \eqref{eq:welchrankin}. The GA produces several high-symmetry frames: ETFs at $n=4,5,7,8,13,16$, the maximal set of MUB at $n=20$, and a $3$-design at $n=40$. The vectors of this frame for $n=40$ correspond to the diagonals of the Witting polytope, a four-dimensional regular complex polytope with $240$ vertices \cite[Example 6]{hoggar1982tdesings}.

It is interesting to observe the creation and destruction of certain structures as $n$ varies. After the set of MUB at $n=20$, particularly after $n=23$, the plot evidences that the vector space starts to be overcrowded with vectors. The frames from $n=24$ to $n=28$ exhibit very little structure in their overlaps. However, from $n=29$ we start to see a gradual sorting of the overlaps that culminates for $n=40$. Similarly, the transition from the SIC-POVM for $n=16$ to the set of MUB for $n=20$ is also interesting. When an additional vector is added from $n=16$ to $n=17$, $4$ vectors reposition to form one of the bases that will later constitute the set of MUB at $n = 20$. A similar phenomenon seems to occur at $n=23$, but it does not manifest in the following values of $n$.

\begin{figure}
    \centering
    \includegraphics[width = 0.65\textwidth]{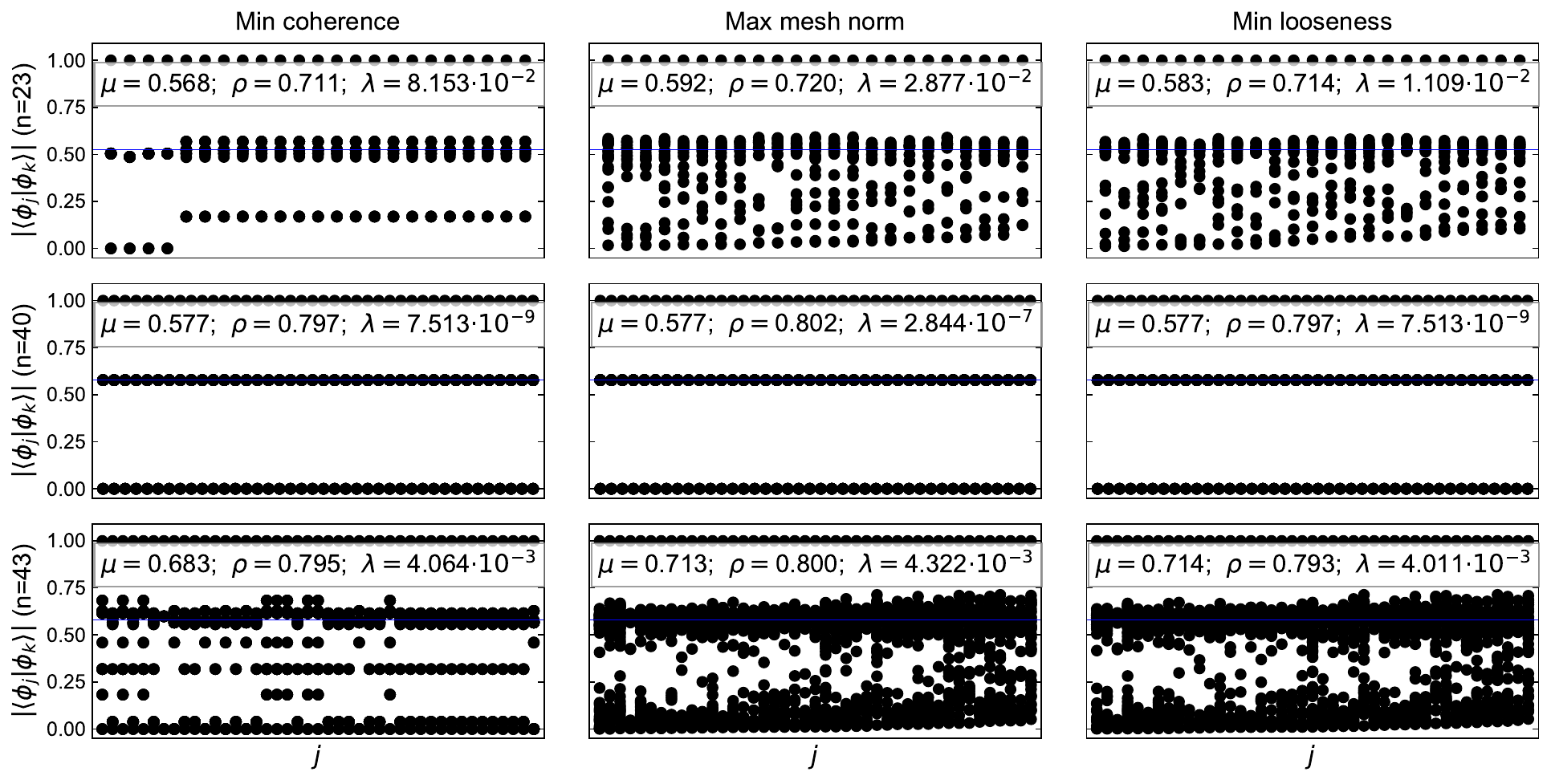}
    \caption{Genetic algorithm stochastic results. Overlaps among vectors of the maximally orthogonal frames for $d=4$ and $s=8$ and $n=23, 40, 43$. Each column shows the result of a different run of the GA that optimizes a different metric. The horizontal blue line indicates the Levenstein bound \eqref{eq:leven}. The coherence $\mu$, mesh norm $\rho$, and looseness $\lambda$ of each frame are indicated in a box within each panel.}
    \label{fig:overlaps_metrics}
\end{figure}

Given the stochastic nature of the GA, it is possible to perform several runs of the algorithm for the same configuration of $d$, $n$ and energy function, and analyze the difference between the resulting frames in terms of coherence, looseness and uniformity. In Fig.~\ref{fig:overlaps_metrics} we show the overlap structure of three different configurations, for $d=4$ and $n=23, 40, 43$. For each configuration we plot the frame with lowest coherence, highest mesh norm (most uniform) and lowest looseness (closest to being tight). For $n=40$, given the high symmetry of the $3$-design, there are no appreciable differences between the three cases. In contrast, for $n=23$ and $n=43$ there are clear differences in the overlap structure between the frames that optimize each metric. Frames that minimize the coherence tend to display smaller angle sets, the overlaps are arranged in well-defined rows. On the other hand, frames that maximize uniformity or minimize looseness have a more scattered distribution of overlaps.

These results indicate that the GA correctly produces universal MOFs, such as ETFs and sets of MUB, for the values of $n$ that they exist. This suggests that the frames obtained for other values of $n$ are also optimal or near optimal, as indicated by the fact that they show coherences that are close to the applicable bound and are in some cases tight. See App. \ref{app:share} for a complementary discussion of these results based on the alternative visualization method in terms of the share of the vector space of each vector. The GA produces similar results in other dimensions (d=$2, 3, 5, 6$), not shown here for conciseness. For $d=7$ the GA fails to provide the SIC-POVM at $n = 49$, indicating a degradation of performance for higher dimensions and frame sizes.

\subsection{The Game of Sloanes: minimizing coherence} 

As discussed in Sec. \ref{subsec:optimalpacking} Grassmannian frames are optimal packings in complex projective space. In Fig. \ref{fig:GoSbenchmark} we compare the coherence of the MOFs produced by the GA for the Riesz $s$-energy with $s = 2d$ with the coherence of the best known packings listed on the Game of Sloanes website \cite{gameofsloanes}. We also plot the applicable bound as a reference of optimal coherence. For the results of the GA we plot two series: one corresponding to the frame produced by the GA for each configuration of $d$ and $n$, and the other obtained with a technique referred to as AUTO in the Game of Sloanes website \cite{gameofsloanes}. This technique consists of obtaining a frame for a given $d$ and $n$ by removing a vector from the best known packing for $d$ and $n+1$. In certain cases, the resulting frame has lower coherence than the frame that can be obtained by running an optimization algorithm such as the GA for $d$ and $n$ directly. Note that the series corresponding to the Game of Sloanes also relies on the AUTO technique. 

In Fig. \ref{fig:GoSbenchmark} we can see several instances where AUTO yields better coherence than a dedicated run of the GA, these tend to occur for values of $n$ preceding a configuration for which the bound is saturated, i.e. before optimal Grassmannian frames. For example, for $d=4$ and $n=24, \ldots, 39$, $d=5$ and $n=26, \ldots, 29$, or $d=6$ and $n=31, \ldots, 35$. As discussed in Sec. \ref{subsubsec:univarsal_mof}, for certain configurations of $d$ and $n$, such as $d=4$ and $n=40$, $d=5$ and $n=30$, and $d=6$ and $n=36$, there exist universal MOFs that are simultaneously Grasmmannian, tight and uniform. In these cases, the minimization of the Riesz $s$-energy with $s=2d$ produces an optimal Grassmannian frame that saturates the applicable bound. For other configurations, minimizing the energy and minimizing coherence are generally incompatible objectives, and as a result the frames produced by the GA are far from the coherence bound and the best known packings. This is supported by the fact that increasing the value of $s$, such that the energy minimized by the GA is a better proxy for coherence, results in frames with lower values of coherence. With some careful finetuning this may allow the GA to even match or beat the best-known packings posted on the Game of Sloanes. However, in our tests, we have not managed to achieve this. For very large values of $s$ the performance of the algorithm worsens. This is probably caused by a change in the energy landscape that makes it less amenable for global optimization by the GA. 

Besides the discrepancy between energy and coherence, there may be an intrinsic factor behind the success of the AUTO technique. Some of the best known packings are also obtained with the AUTO technique, such as $d=4$ and $n=26, \ldots, 39$, $d=5$ and $n=26, \ldots, 29$, or $d=6$ and $n=32, \ldots, 35$ \cite{gameofsloanes}. This is the case even though SOTA algorithms such as the ones discussed in Sec. \ref{optimalpackings} are geared towards minimizing coherence. It could be the case that this is just another manifestation of the discrepancy between energy and coherence, because even SOTA algorithms do not minimize the coherence but a smooth version of it that incorporates all overlaps and not just the maximal ones. However, if we assume that the energy minimized by SOTA algorithms is sufficiently similar to the coherence, the success of AUTO seems to indicate that the energy landscape becomes extremely complex for certain configurations, to the point that it is easier to find frames with lower coherence for configurations with a larger number of vectors, and then remove vectors to arrive to a frame for the desired configuration. The most extreme example of this is $d=4$ and $n=27$. One achieves a frame with lower coherence by first adding $13$ vectors and solving for $n=40$ with the GA or a SOTA algorithm, and then removing $13$ vectors from the resulting optimal Grassmannian frame, rather than solving directly for $d=4$ and $n=27$.

Furthermore, we wonder if the orthoplex bound, which proves that frames obtained with AUTO by removing vectors from maximal sets of MUB are optimal Grassmannian frames, may be generalizable to any optimal Grassmanian frame that saturates the Levenstein bound. In other words, if for every MOF $\bar \Phi_{d, n}$ such that $\mu(\Phi_{d, n})$ saturates the Levenstein bound, for example $d=4$ and $n=40$, there exists a bound with value $\mu(\Phi_{d, n})$ and applicable for some $m < n$.  The existence of such a bound would prove that some or all of the configurations obtained with AUTO in those cases, such as $d=4$ and $n=27, \ldots, 39$, are not only the best known packings, but actually optimal.

Finally, note that AUTO-generated frames are necessarily suboptimal in terms of uniformity. If we assume that the starting frame is uniform, removing vectors leaves a gap in the vector space, i.e. an area with poor coverage by vectors of the frame. This is analogous to the phenomenon illustrated in Fig. \ref{fig:cohdeg} and discussed in Sec. \ref{subsec:properties_MOF} for $d=2$ and $n=5$. Note that the frame in column a) of Fig. \ref{fig:cohdeg} has been obtained using AUTO from the optimal Grassmannian frame for $n=6$, which corresponds to the maximal set of MUB. Likewise, a redistribution of vectors without affecting the maximal overlaps should be possible whenever AUTO is used, improving the uniformity of the frame without affecting coherence. When a maximal set of MUB exists, this is the case for frames saturating the orthoplex bound, for $d^2 < n < d(d+1)$, which are obtained with AUTO from the corresponding set of MUB at $n=d(d+1)$.

\begin{figure}
    \centering
    \includegraphics[width = 0.8\textwidth]{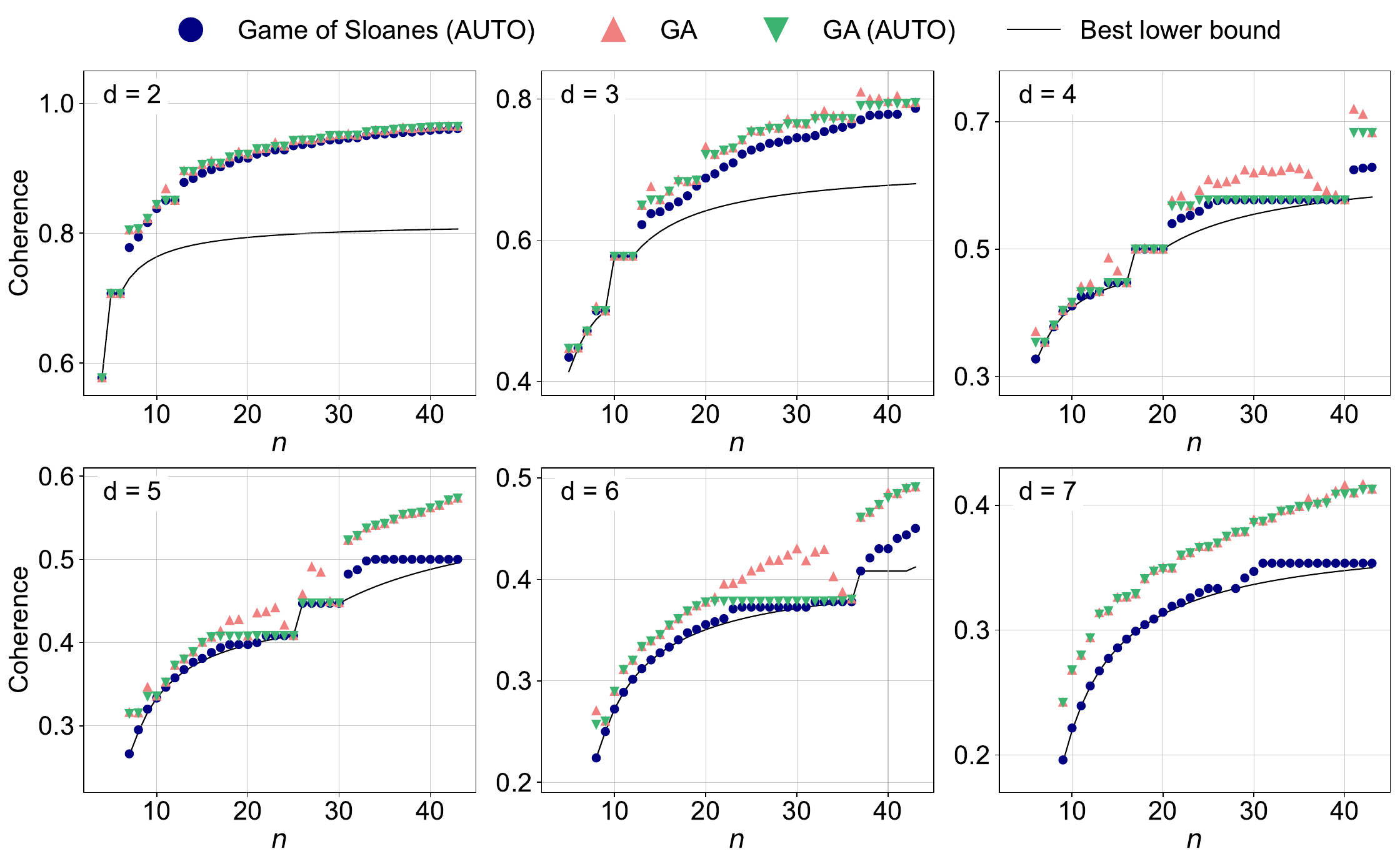}
    \caption{Comparison with the Game of Sloanes. Each panel shows the coherence of the best frames obtained out of ten runs of the GA and of the best known packings extracted from the Game of Sloanes \cite{gameofsloanes} for different dimension $d$ and sizes $n=d, \ldots, 43$. The solid black line marks the applicable coherence bound.}
    \label{fig:GoSbenchmark}
\end{figure}

\section{Conclusions}
\label{conclusions}

In this paper we have defined maximally orthogonal frames and discussed their relation with popular frame families such as tight or Grasmmannian frames. We have demonstrated that there is not a unique definition of maximally orthogonal frames, since they depend on the choice of energy function to weigh the different pairwise overlaps. As a consequence, we have established that obtaining each type of maximally orthogonal frame is a multivariate global minimization problem (of the corresponding energy function) that can be tackled numerically and discussed its connection with other related numerical problems such as the Thomson problem, the finding of optimal packings in complex projective space and the finding of optimal reference states for quantum classifiers. Then, we have presented a genetic algorithm that can generate maximally orthogonal frames of any dimension and size. The genetic algorithm uses a hybrid approach with local optimization at each generation. We have specifically designed the recombination and mutation steps to preserve the local structure of the parent frames.

We have devised visualization methods for high-dimensional frames and established that the GA is able to produce highly-symmetric universal frames such as equiangular tight frames, SIC-POVMs and maximal sets of mutually unbiased bases. In addition, we have assessed the performance of the GA by applying it to the Thomson problem and the problem of finding optimal packings in complex projective space. We find that the GA produces optimal or near-optimal frames in both cases. Despite some attempts to finetune the hyperparameters of the genetic algorithm, we have not managed to beat SOTA algorithms for these problems. We believe this reflects the algorithm’s simple design focused on demonstrating feasibility rather than optimizing for peak performance.

The genetic algorithm may be upgraded by implementing a non-deterministic selection strategy with an adaptive rate to exercise better control over the convergence dynamics and achieve a better exploration of the state space. It is also possible that a better choice of energy function unlocks new performance gains for certain metrics such as the coherence. Finally, given the success of the GA it may be worth exploring the use of SOTA evolutionary strategies such as the covariance matrix adaptation evolution strategy (CMA-ES) to generate maximally orthogonal frames.

The code for this paper can be found in Ref. \cite{code}.

\section*{Acknowledgements}

We are grateful to David Zueco for his encouragement throughout the project, and to Jes\'us Carrete for his careful reading of the manuscript.
This research project was made possible through the access granted by the Galician Supercomputing Center (CESGA) to its supercomputing infrastructure. The supercomputer FinisTerrae III and its permanent data storage system have been funded by the Spanish Ministry of Science and Innovation, the Galician Government and the European Regional Development Fund (ERDF).
We acknowledge funding through Grant No. CEX2023-001286-S, from 
MCIN/AEI/10.13039/501100011033 and the EU NextGenerationEU/PRTR, and No. TED2021-131447B-C21,
from MCIN/AEI/10.13039/501100011033.  This work was also supported by the Spanish Ministry for Digital Transformation and of Civil Service of the Spanish Government through the QUANTUM ENIA project call - Quantum Spain, EU through the Recovery, Transformation and Resilience Plan – NextGenerationEU within the framework of the Digital Spain 2026.
We also acknowledge the Gobierno de Arag\'on (Grant No. E09-17R Q-MAD), Quantum Spain and the CSIC Quantum Technologies Platform PTI-001.
S. R-J.  acknowledges financial support from Gobierno de Arag\'on through a doctoral fellowship.
J. R-R acknowledges support from the Ministry of Universities of the Spanish Government through the grant FPU2020-07231.

\appendix

\section{Numerical details}
\label{app:numerics}

\subsection{Scaling and runtime}
We find that the GA converges in between $6$ and $30$ generations, depending on the configuration. It tends to converge faster for configurations in which it is possible to find a highly symmetric frame such as an ETF. For instance, for $d=4$ and $n=16$ it finds the SIC-POVM in $5$ generations. 

Runtimes depend on both the dimension $d$ and the frame size $n$. The state vector size for the local optimization scales as $(n-1) \times d$, and the computation of the energy scales as $d \times n \times (n-1)$. For reference, we provide in Table \ref{tab:runtimes} the runtimes on an M1 Mac Mini for different configurations, with the number of generations of the GA capped at $15$.
\begin{table}[]
\begin{tabular}{@{}rrrrr@{}}
\toprule
\multicolumn{1}{c}{$d$} & \multicolumn{1}{c}{$n$} & \multicolumn{1}{c}{Runtime}   & \multicolumn{1}{c}{Generations} & \multicolumn{1}{c}{Time/generation} \\ \midrule
$2$                     & $6$                     & $0.6 \, {\rm s}$              & $5$                                               & $0.12 \, {\rm s}$                   \\
$2$                     & $24$                    & $99.9 \, {\rm s}$             & $15$                                              & $6.6 \, {\rm s}$                    \\
$4$                     & $16$                    & $28.7 \, {\rm s}$             & $7$                                               & $4.1 \, {\rm s}$                    \\
$4$                     & $40$                    & $9 \, {\rm min} \, 6 \, {\rm s}$ & $9$                                               & $1 \, {\rm min} \, 7 \, {\rm s}$       \\ \bottomrule
\end{tabular}
\label{tab:runtimes}
\caption{Runtimes of the genetic algorithm for different configurations. The number of generations was capped at $15$.}
\end{table}

\subsection{Computing the mesh norm and the share of the vector space}
The mesh norm and the share of the vector space can only be computed approximately by performing a finite sampling of $N_H$ vectors from the Haar distribution for each dimension $d$. 
In Fig. \ref{fig:haar_convergence} we show how the mesh norm depends on $N_H$ for different configurations.  
We find that a value of $N_H = 10^7$ is sufficient to resolve changes in the mesh norm $\geq 10^{-2}$, which is sufficient to compare, e.g., the frames obtained by optimizing different metrics in Fig. \ref{fig:overlaps_metrics}.

\begin{figure}
    \centering
    \includegraphics[width=0.5\linewidth]{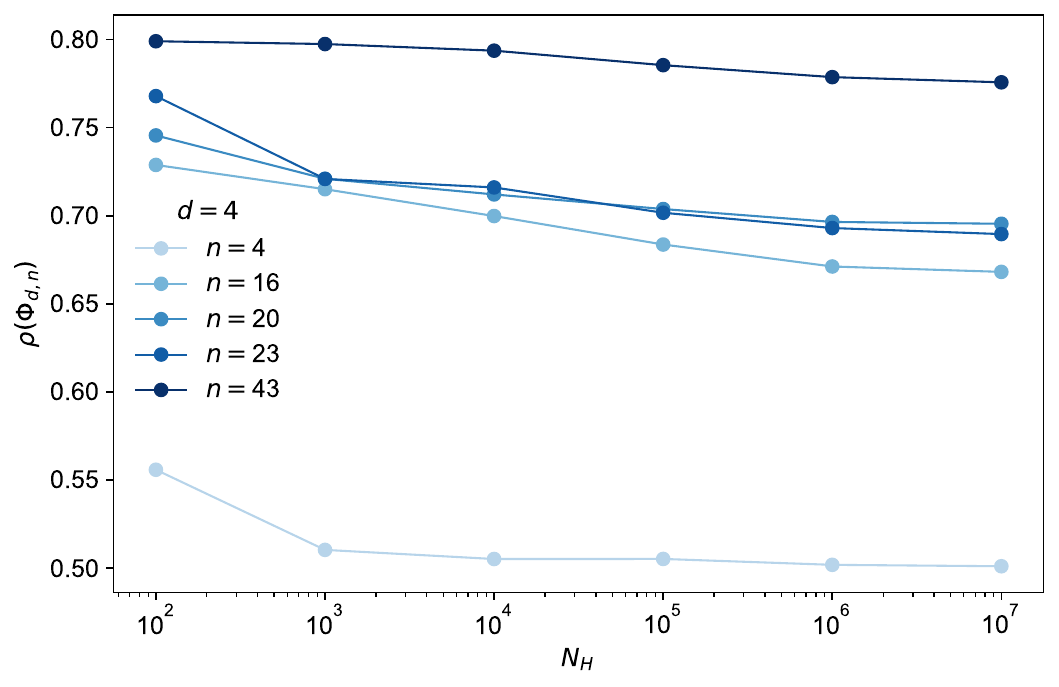}
    \caption{Dependence of the mesh norm with the number of Haar vectors $N_H$ for $d=4$ and different values of $n$.}
    \label{fig:haar_convergence}
\end{figure}

\section{Visualizing uniformity: share of the vector space}
\label{app:share}

In Sec. \ref{subsec:metricsbounds} we defined the normalized share of the vector space $\beta_j$ as a way to visualize the uniformity of a frame. By definition, the share and the mesh norm are related. The mesh norm measures the radius of the largest gap left in the vector space by vectors of the frame. The share measures the portion of the vector space that is closest to each vector of the frame. The relation stems from the fact that the vectors of the frame surrounding the largest gap will have a larger share. Therefore, the mesh norm is positively correlated with the largest values of the share. This relation is similar to the relation between coherence and overlaps, where coherence is the maximum overlap. Note also that the share of the vector space is the metric to optimize when looking for optimal reference states for a quantum classifier (See Sec. \ref{quantumclassifier}). A non-uniform distribution of the shares is associated with default biases of the quantum classifier.

In Figs. \ref{fig:shares_bestE} and \ref{fig:shares_metrics} we visualize the uniformity of the frames obtained with the GA for $d=4$. Figure \ref{fig:shares_bestE} is complementary to Fig. \ref{fig:overlaps_bestE}. In Fig. \ref{fig:shares_bestE} we show the frames obtained for $d=4$ and $n \in [4, 43]$. Each panel displays the share of the vector space $\beta_j$ of each vector $|\phi_k\rangle\in\bar\Phi_{4,n}[s=8]$. In addition, the standard deviation of the shares is indicated within each panel. Uniform frames have $\beta_j \approx 1$ for all vectors and a small standard deviation. It is interesting to note that all frames up to the SIC-POVM for $n=16$ are relatively uniform. Appreciable non-uniformities only appear beyond this point. Nevertheless, highly symmetric frames such as the set of MUB or the $3$-design at $n=40$ are again relatively uniform.

\begin{figure}
    \centering
    \includegraphics[width = \textwidth]{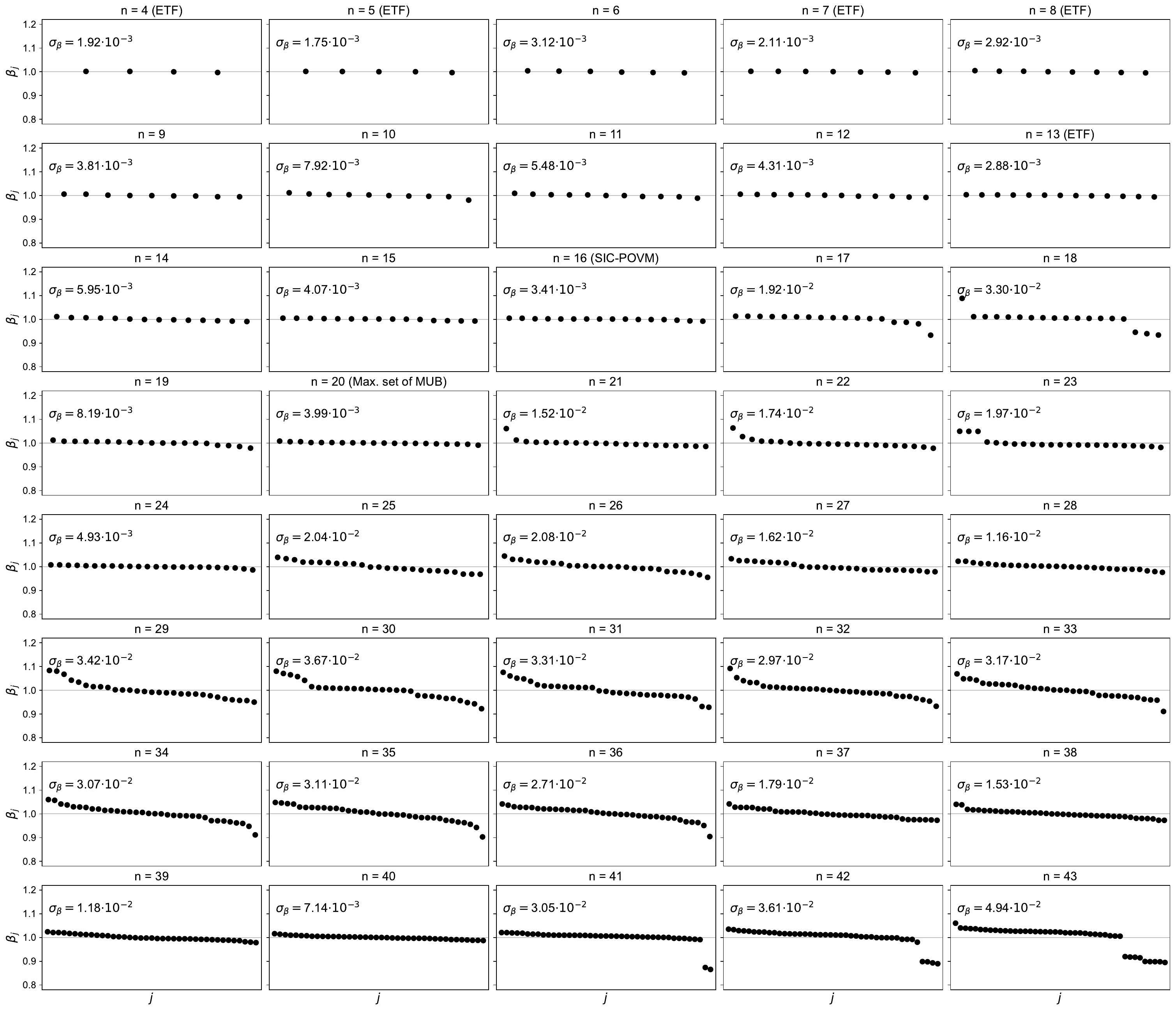}
    \caption{Maximally orthogonal frames for $d=4$ and $s=8$. Each panel shows the normalized share of the vector space $\beta_j$ \eqref{eq:share_hilbert} of each vector the best MOF out of ten runs of the GA for a certain size $n$. The horizontal black line is a guide to the eye marking the value $1$. The panel title indicates whether it is a high-symmetry frame such as an ETF, maximal set of MUB or SIC-POVM. The standard deviation of the share $\sigma_\beta$ is indicated within each panel.}
    \label{fig:shares_bestE}
\end{figure}

Figure \ref{fig:shares_metrics} is complementary to Fig. \ref{fig:overlaps_metrics}. In Fig. \ref{fig:shares_metrics} we plot the share of the vector space of three different configurations, for $d=4$ and $n=23, 40, 43$. For each configuration we plot the frame with lowest coherence, highest mesh norm and lowest looseness (closest to being tight). For $n=40$, given the high symmetry of the $3$-design, there are barely any differences between the three cases. In contrast, for $n=23$ and $n=43$ there are clear differences in the shares between the frames that optimize each metric. Frames that minimize coherence tend to display shares grouped in two or more distinct values. On the other hand, frames that maximize the mesh norm or minimize looseness have more evenly distributed shares. This is reflected in the standard deviation of the share, which is higher for the frames that minimize coherence.

We note that for several configurations in Figs. \ref{fig:shares_bestE} and \ref{fig:shares_metrics} the distribution of shares is mostly concentrated around $1$, with one or a few outliers with a value significantly different than $1$. In these cases, the standard deviation may not reflect the actual uniformity of the frame and it could be interesting to incorporate estimators of kurtosis or multimodality of the distribution of shares.

\begin{figure}
    \centering
    \includegraphics[width = 0.65\textwidth]{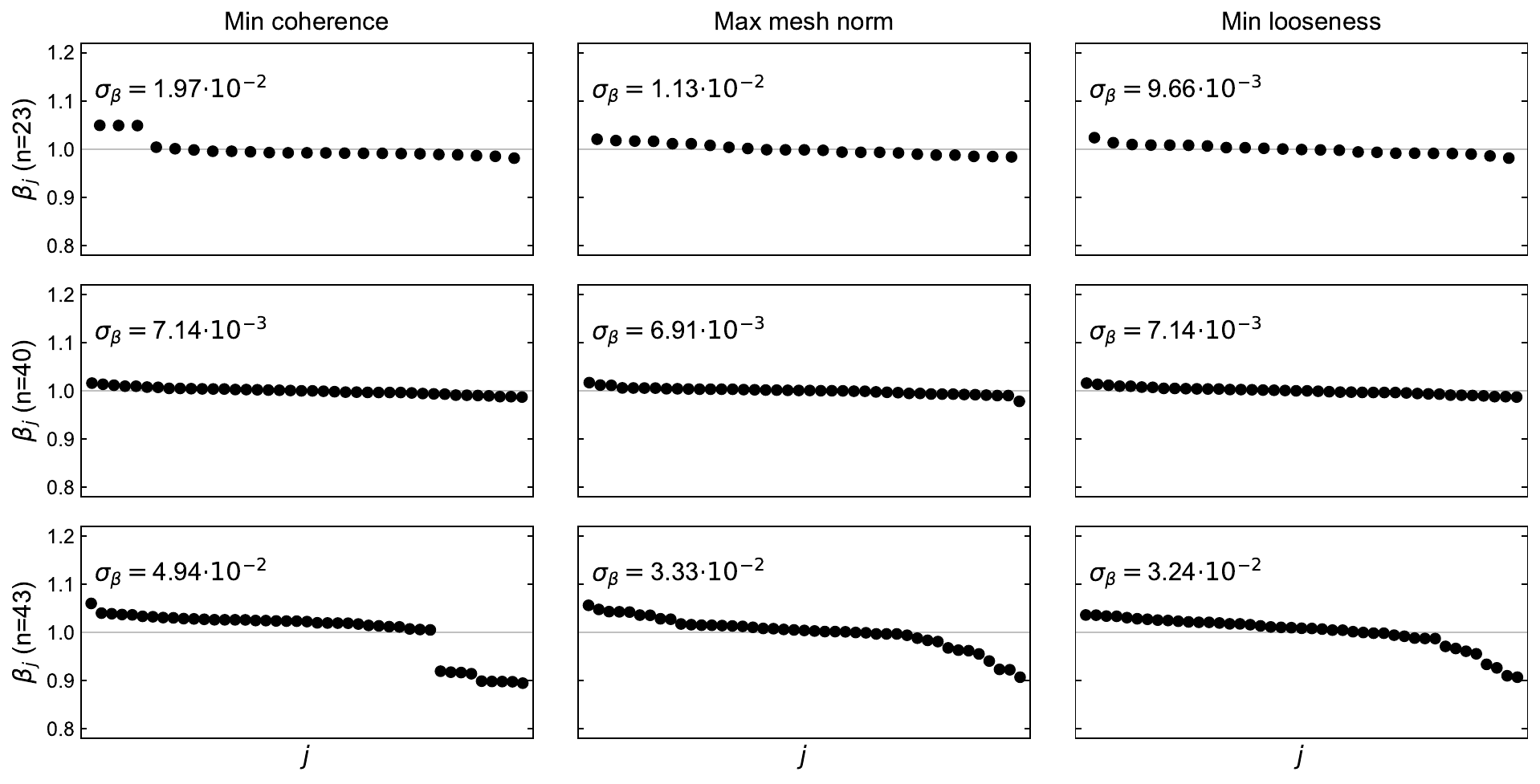}
    \caption{Genetic algorithm stochastic results. Normalized share of the vector space $\beta_j$ \eqref{eq:share_hilbert} of the maximally orthogonal frames for $d=4$ and $s=8$ and $n=23, 40, 43$. Each column shows the result of a different run of the GA that optimizes a different metric. The horizontal black line is a guide to the eye marking the value $1$. The standard deviation of the share $\sigma_\beta$ is indicated within each panel.}
    \label{fig:shares_metrics}
\end{figure}

\bibliography{main.bib}

\end{document}